%% file: main.tex
    \documentclass[12pt,a4paper]{article}
\input{settings}

\begin{document}

\title{No-Rainbow Problem and the Surjective Constraint Satisfaction Problem}
\author{Dmitriy Zhuk} 

\maketitle

\begin{abstract}
The Surjective Constraint Satisfaction Problem (SCSP) is the problem 
of deciding whether there exists a surjective assignment to a set of variables
subject to some specified constraints, 
where a surjective assignment is an assignment containing all elements of the domain.
In this paper we show that 
the most famous SCSP, called No-Rainbow Problem, is NP-Hard. 
Additionally, we disprove the conjecture saying that 
the SCSP over a constraint language $\Gamma$ and 
the CSP over the same language with constants
have the same computational complexity up to poly-time reductions.
Our counter-example also shows
that the complexity of the SCSP cannot be described in terms of polymorphisms of the constraint language.
\end{abstract}

\input{main_text}

\input{counterexample}

\bibliographystyle{plain}
\bibliography{refs}

\end{document}

%% file: settings.tex
    \usepackage[cp1251]{inputenc}
    \usepackage[english]{babel}
\usepackage{algorithm}
\usepackage[noend]{algpseudocode}    
    \usepackage{amsmath,amssymb,amsthm}
         \usepackage{graphicx}
     \usepackage{wrapfig}
\usepackage{tikz}
\usepackage{xypic}
\usetikzlibrary{shapes}

  \DeclareMathOperator\CSP{CSP}
  \DeclareMathOperator\SCSP{SCSP}

  \DeclareMathOperator\Inv{Inv}
  \DeclareMathOperator\Image{Im}
  \DeclareMathOperator\Pol{Pol}

  \DeclareMathOperator\SurjHom{SurjHom}

\theoremstyle{definition}
\theoremstyle{plain}
\newtheorem{thm}{Theorem}
\newtheorem{problem}{Question}
\newtheorem{conj}{Conjecture}
\newtheorem{lem}[thm]{Lemma}

%% file: main_text.tex
\section{Introduction}
\subsection{Constraint Satisfaction Problem}
The \emph{Constraint Satisfaction Problem (CSP)} is the problem of deciding whether there is an assignment to a set of variables
subject to some specified constraints. 
Formally, it can be defined in the following way.
Let $A$ be a finite set, 
$\Gamma$ be a set of relations on $A$, called 
a \emph{constraint language}.
Then \emph{the Constraint Satisfaction Problem over the constraint language} $\Gamma$, denoted by $\CSP(\Gamma)$, is
the following decision problem:
given a formula of the form 
$$R_{1}(z_{1,1},\dots,z_{1,n_1})\wedge
\dots
\wedge R_{s}(z_{s,1},\dots,z_{s,n_s}),
$$
where $R_{1},\dots,R_{s}\in\Gamma$ and each $z_{i,j}\in\{x_1,\dots,x_{n}\}$;
decide whether the formula is satisfiable.
It is well known that many combinatorial problems can be expressed as $\CSP(\Gamma)$
for some constraint language $\Gamma$.
Moreover, for some sets $\Gamma$ the corresponding decision problem can be solved in polynomial time;
while for others it is NP-complete.
It was conjectured in 1998 that
$\CSP(\Gamma)$ is either in P, or NP-complete \cite{FederVardi}.
In 2017 this conjecture was resolved independently 
by Andrei Bulatov \cite{BulatovProofCSP,bulatov2017dichotomy} and Dmitriy Zhuk \cite{ZhukFOCSCSPPaper,MyProofCSP}. Moreover, the classification of the complexity for different constraint languages turned out to be very simple and was given in terms of polymorphisms.
We say that a $k$-ary operation $f$ is a \emph{polymorphism} of
an $m$-ary relation $R$
if 
whenever $(x^1_1,\ldots,x^m_1),\ldots,(x^1_k,\ldots,x^m_k)$ in $R$, then also $(f(x^1_1,\ldots,x^1_k),\ldots,f(x^m_1,\ldots,x^m_k))$ in $R$.
In this case we also say that $f$ \emph{preserves} $R$.
An operation $f$ is a \emph{polymorphism} of $\Gamma$ if it is a polymorphism of every relation in $\Gamma$.
An operation $w$ is called a \emph{weak near-unanimity (WNU) operation} 
if it satisfies the following identities
$$
w(y,x,\dots,x)=
w(x,y,x,\dots,x)=\dots = w(x,\dots,x,y).
$$
For instance a majority operation, conjunction, disjunction are WNU operations.
Then the complexity of $\CSP(\Gamma)$ is described by the following theorem.

\begin{thm}[\cite{BulatovProofCSP,bulatov2017dichotomy,ZhukFOCSCSPPaper,MyProofCSP}]\label{FederVardiConj}
Suppose $\Gamma$ is a finite constraint language on a finite set $A$.
Then $\CSP(\Gamma)$ is solvable in polynomial time 
if there exists a WNU polymorphism of $\Gamma$;
$\CSP(\Gamma)$ is NP-complete otherwise.
\end{thm}

\subsection{Surjective Constraint Satisfaction Problem}

In this paper we consider a variant of the CSP with an additional 
global constraint saying that a solution should be surjective.
\emph{The Surjective Constraint Satisfaction Problem over a constraint language $\Gamma$ on a domain $A$}, denoted by $\SCSP(\Gamma)$, is
the following decision problem:
given a formula of the form 
$$R_{1}(z_{1,1},\dots,z_{1,n_1})\wedge
\dots
\wedge R_{s}(z_{s,1},\dots,z_{s,n_s}),
$$
where $R_{1},\dots,R_{s}\in\Gamma$ and each $z_{i,j}\in\{x_1,\dots,x_{n}\}$;
decide whether there exists a solution
such that $\{x_{1},\dots,x_{n}\}=A$.
Here we assume that the set of variables 
$\{x_1,\dots,x_n\}$ is fixed and do not require all the variables 
to appear in some constraint.
Unlike the complexity of the CSP, 
the complexity of $\SCSP(\Gamma)$ remains unknown even for very 
simple constraint languages $\Gamma$.

\subsection{Surjective Graph Homomorphism Problem}

Probably, the most natural examples of the Surjective CSP 
are defined as the graph homomorphism problem.
Assume that a graph $H$ is fixed, 
\emph{the Surjective Graph Homomorphism Problem}, denoted by $\SurjHom(H)$, is the problem of deciding 
for a given graph $G$ whether there exists a surjective 
homomorphism from $G$ to $H$. 
This problem is also known as 
\emph{the Vertex-Compaction Problem}
\cite{DBLP:journals/jda/Vikas18}
or \emph{the Surjective $H$-Colouring Problem} \cite{DBLP:journals/computability/GolovachJMPS19}.
For usual homomorphism this problem is known as
\emph{the $H$-colouring problem}.
Note that $\SurjHom(H)$ is equivalent 
to $\SCSP(\{H\})$, where the graph $H$ is viewed as a binary relation.

An interesting fact about the complexity of $\SurjHom(H)$ is that 
it remained unknown for many years even for very simple graphs $H$. 
Recall that the situation was different with the complexity of the CSP and the $H$-colouring problem: even though the general classification remained open for many years, nobody 
could show a simple constraint language or a graph with unknown complexity.
We will discuss two popular examples of graphs with unknown complexity of the $\SurjHom$ in more detail.

First example is the complexity of $\SurjHom$ 
for the reflexive 4-cycle
(undirected having a loop at each vertex), which was formulated 
as an open question in \cite{dantas2005finding} and 
\cite{fleischner2009covering}, later
it was a principal open question in 
\cite{ito2011parameterizing,ito2011disconnected}, and the central question discussed in \cite{cook20102k2} and \cite{dantas20122k2}. 
This problem is known as \emph{the disconnected cut problem}, since it is equivalent to finding a cutset
(a set whose removal
results in a disconnected graph) such that the cutset is disconnected itself,
and it has attracted a lot of interest from the graph theory community. 
This problem was 
known 
to be tractable for many graph classes \cite{cook20102k2,dantas20122k2,fleischner2009covering,ito2011parameterizing}, 
but in 2011 it was finally proved that  the disconnected cut problem is NP-complete \cite{DBLP:journals/jct/MartinP15}.

Second long-standing problem is 
the complexity of $\SurjHom$ 
for the non-reflexive 6-cycle (undirected without loops), 
that has been of interest since 1999 \cite{vikas1999computational}
when Narayan Vikas proved 
NP-completeness of a similar problem for the 6-cycle, 
called \emph{the compaction problem}.
The difference between the compaction problem and the vertex-compaction problem is that for compaction we 
require the homomorphism to be edge-surjective.
In 2017 it was finally proved 
that $\SurjHom$ for the 6-cycle
is also NP-complete \cite{DBLP:conf/mfcs/Vikas17}.
It is worth mentioning that 
the compaction and the vertex-compaction problems 
have the same complexity for all known graphs, and it was conjectured by Winkler, Vikas, and others that these problems are polynomially equivalent.
For more information about 
the relationship between 
the compaction problem and the vertex-compaction problem 
(and also the retraction problem) see \cite{DBLP:journals/jda/Vikas18}.

Many other results on the complexity of 
$\SurjHom(H)$ can be found in 
\cite{DBLP:journals/toct/LaroseMP19,DBLP:journals/computability/GolovachJMPS19, DBLP:journals/siamdm/FockeGZ19,VikasOther, VikasIrreflexive,VikasReflexive,DBLP:journals/algorithmica/Vikas13,DBLP:journals/tcs/GolovachPS12,DBLP:conf/mfcs/Vikas17,DBLP:journals/jct/MartinP15}. For instance, 
the complexity of $\SurjHom$ 
is known
for all graphs on at most four vertices \cite{DBLP:journals/computability/GolovachJMPS19}.
But as far as we know the complexity remains unknown for graphs of size 5 and even for cycles, 
which makes this problem very intriguing because of the simplicity of the formulation.


\subsection{The Complexity of SCSP}

Let us discuss what we know about the complexity of 
the SCSP and the CSP.
The complexity of $\CSP(\Gamma)$ is known for any constraint language $\Gamma$
\cite{BulatovProofCSP,bulatov2017dichotomy,ZhukFOCSCSPPaper,MyProofCSP}. The complexity of $\SCSP(\Gamma)$ 
is widely open.
Since we can always add dummy variables we never use and these dummy variables could give surjectivity,  $\CSP(\Gamma)$ can be trivially reduced to 
$\SCSP(\Gamma)$.
Also, it is sufficient to consider the reflexive 4-cycle $H$ to
get a constraint language $\Gamma=\{H\}$ such that 
$\SCSP(\Gamma)$ is NP-complete but 
$\CSP(\Gamma)$ is trivial and solvable in polynomial time.
Thus, sometimes $\SCSP(\Gamma)$ is harder than $\CSP(\Gamma)$.

Another important observation is that 
the Surjective CSP can be reduced to the CSP over the same language with constants (see \cite[Section 2]{SCSPSURVEY}). Adding constants is equivalent to adding 
unary singleton relations $\{\{a\}\mid a\in A\}$ to the constraint language, because using these relations we can write
$x_{i} = a$.
Let us show how this reduction works.

\begin{lem}[\cite{SCSPSURVEY}]\label{SCSPTOCSP}
There exists a polynomial-time Turing reduction from $\SCSP(\Gamma)$ to $\CSP(\Gamma\cup\{\{a\}\mid a\in A\})$
\end{lem}
\begin{proof}
Let $A = \{a_1,\dots,a_{n}\}$. Suppose 
we have an instance $\mathcal I$ of 
$\SCSP(\Gamma)$.
First, we guess $n$ variables $x_{i_1},\dots,x_{i_n}$ such that these variables give all elements of $A$ in a solution, 
that is 
$x_{i_j}=a_j$ for every $j$.
Then we consider the following instance of $\CSP(\Gamma\cup\{\{a\}\mid a\in A\})$
$$\bigwedge\limits_{j=1}^{n}(x_{i_{j}} = a_{j})\wedge \mathcal I.$$
If it has a solution, then $\mathcal I$ has a surjective solution. 

Since we have only polynomially many choices for the 
variables, this gives us a polynomial-time Turing reduction to 
$\CSP(\Gamma\cup\{\{a\}\mid a\in A\})$.
\end{proof}

As a corollary we derive that 
$\CSP(\Gamma\cup\{\{a\}\mid a\in A\})$
and 
$\SCSP(\Gamma\cup\{\{a\}\mid a\in A\})$
have the same complexity.

Recall that the classification of the complexity 
of the CSP for different constraint languages was given in terms of polymorphisms
(see Theorem \ref{FederVardiConj}).
It is natural to assume that the complexity of the SCSP can be 
described in a similar way.
In \cite{HubieSlides} Hubie Chen asked 
whether 
Lemma~\ref{SCSPTOCSP} holds in both directions, which would imply 
the characterization of the complexity of the SCSP in terms of polymorphisms.
Thus, he conjectured the following characterization of the complexity of 
$\SCSP(\Gamma)$.

\begin{conj}\label{HubieConj}
For any constraint language $\Gamma$ 
the problems 
$\CSP(\Gamma\cup\{\{a\}\mid a\in A\})$
and 
$\SCSP(\Gamma)$ 
are polynomially equivalent.
\end{conj}

An operation $f$ is called 
\emph{idempotent} if it satisfies 
$f(x,x,\dots,x) = x$.
Then, using Theorem \ref{FederVardiConj}
Conjecture \ref{HubieConj} can be formulated in the following form.

\begin{conj}\label{HubieConjTwo}
For any constraint language $\Gamma$ the problem $\SCSP(\Gamma)$ is solvable in polynomial time 
if there exists an idempotent WNU polymorphism of $\Gamma$;
$\SCSP(\Gamma)$ is NP-complete otherwise.
\end{conj}

As it was proved in \cite{SCSPforTwo1997} (see also \cite{SCSPforTwoBook}), Conjectures \ref{HubieConj} and \ref{HubieConjTwo}
hold for a 2-element domain. 
Also, all known results on the complexity of the 
SCSP agree with these conjectures (see \cite{SCSPSURVEY}).

Additionally, in \cite{chen2014algebraic} Hubie Chen confirmed that polymorphisms of $\Gamma$ can be used to 
describe the complexity of $\SCSP(\Gamma)$.
For instance, 
he proved NP-hardness of $\SCSP(\Gamma)$ for 
$\Gamma$ admitting only essentially unary polymorphisms.
For more results on the complexity of the SCSP see the survery 
\cite{SCSPSURVEY}.

\subsection{No-Rainbow Problem}

In this paper we consider the most famous constraint language with 
unknown complexity of the SCSP.
Let $N =\{(a,b,c)\in\{0,1,2\}^{3}\mid \{a,b,c\}\neq\{0,1,2\}\}$. 
The problem $\SCSP(\{N\})$ is called \emph{No-Rainbow Problem}
because 
if we interpret 
$0,1$, and 2 as colors then the relation $N$ forbids  
``rainbow''.
Apparently, No-Rainbow Problem was first formulated  under this name in \cite{bodirsky2004disser}.
Nevertheless, this problem can be viewed as the problem of 
coloring of a co-hypergraph in exactly 3 colors,
and such problems were studied in many papers earlier.
The notion of coloring of a mixed hypergraph appeared in \cite{voloshin1993mixed}. 
Sufficient and necessary conditions for the existence of a
$k$-coloring of a co-hypergraph were studied in \cite{diao2000upper}.
The complexity of the problem of such coloring was studied in
\cite{ColoringMixedHypertrees}. 

Despite the simplicity of the formulation, the complexity of $\SCSP(\{N\})$ 
was an open question for many years and was formulated 
many times as an important open problem 
\cite{SCSPSURVEY,HubieSlides}.
For example, Hubie Chen presented this 
problem as one of three concrete open questions 
at several conferences (see \cite{HubieSlides})
emphasizing that this problem is an important step 
toward a full classification of the complexity.

Note that the CSP 
over $N$ is trivial because 
any instance has a trivial solution 
where all variables are equal.
But if we add all constant relations then 
the problem becomes 
NP-hard, that is 
$\CSP(\{N,\{0\},\{1\},\{2\}\})$ is NP-complete \cite{BulatovForThree}. 

As we mentioned earlier 
$\SCSP(\{N,\{0\},\{1\},\{2\}\})$
and $\CSP(\{N,\{0\},\{1\},\{2\}\})$ have the same complexity.
It is interesting to compare the complexity of 
the CSP and the SCSP if we add some constant relations but probably not all of them.
$$\begin{array}{|c|c|c|}
\hline
  \text{Constraint Language} & \CSP & \SCSP \\ \hline
\{N\}  &\text{P (Trivial)} & ???\\
\{N,\{0\}\}  &\text{P (Trivial)} & ???\\
\{N,\{0\},\{1\}\}  &\text{P (Trivial)} & \text{NP-comp. \cite{SCSPSURVEY}}\\
\{N,\{0\},\{1\},\{2\}\}& \text{NP-comp.}& \text{NP-comp.}\\
\hline
\end{array}$$


For instance, as it was shown in \cite{SCSPSURVEY}
if we add just two constant relations then 
the CSP is still trivial but 
the SCSP is already NP-complete. 
It can be shown that the only problems in the table with unknown complexity, that is  
$\SCSP(\{N\})$ and $\SCSP(\{N,\{0\}\})$, are polynomially equivalent. 
In fact, 
to reduce $\SCSP(\{N,\{0\}\})$ to $\SCSP(\{N\})$ 
we just 
replace all variables $x_i$ that appear in a constraint $x_i = 0$ by a unique fresh variable $z$ and remove all constraints of the form 
$x_i = 0$.

 \section{Main Results}\label{MainResultsSection}

In this section we formulate two main results of this paper. 

First, we proved that the No-Rainbow Problem is NP-hard, 
and therefore it is NP-complete.

\begin{thm}\label{NoRainbowTHM}
$\SCSP(\{N\})$ is NP-complete.
\end{thm}
Recently, Hubie Chen published a 
paper explaining the algebraic framework for proving such hardness results \cite{chen2020algebraic}.
He shows that this approach 
can also be used
to prove NP-hardness 
for the reflexive 4-cycle (the disconnected cut problem)
and for all NP-hard constraint languages on a 2-element domain.

Theorem \ref{NoRainbowTHM} agrees with Conjectures \ref{HubieConj} and \ref{HubieConjTwo}.
Unfortunately, we found a counter-example to these conjectures, which is the second main result of this paper.
Let us define 5-ary and 3-ary relations on the set $\{0,1,2\}$
(here columns are tuples):
$$R = 
\begin{pmatrix}
0&0&0&0&0&0\\
2&2&2&1&1&1\\
2&1&1&1&1&1\\
1&2&2&1&1&1\\
0&1&2&0&1&2
\end{pmatrix}, R' = 
\begin{pmatrix}
2&2&1\\
2&1&1\\
1&2&1
\end{pmatrix}.$$
Note that $R'$ is the projection of $R$ onto coordinates 2, 3 and 4.
Hence, every polymorphism of $R$ is also a polymorphism of 
$R'$.

\begin{thm}\label{CEthm} We have
\begin{enumerate}
    \item[(1)] $\CSP(\{R,\{0\},\{1\},\{2\}\})$ is NP-hard;
    \item[(2)] $\SCSP(\{R\})$ is solvable in polynomial time;
    \item[(3)] $\SCSP(\{R'\})$ is NP-Hard.
\end{enumerate}
\end{thm}

Thus, items (1) and (2) of 
Theorem \ref{CEthm} disprove
Conjectures \ref{HubieConj} and \ref{HubieConjTwo}. 
Also, comparing (2) and (3) we derive that 
the complexity of $\SCSP(\Gamma)$ cannot be described in terms of polymorphisms.
Hence, we need a brand new idea to describe the complexity of the SCSP for all constraint languages.

Nevertheless, our counter-example 
is of arity 5 and it is important for the proof.
Thus, probably for binary relations 
Conjectures \ref{HubieConj} and \ref{HubieConjTwo} still hold.
Another way to break our counter-example is 
to add all projections of relations from the constraint language to the constraint language. 
For instance, we could achieve this situation by considering only $\Gamma$ that are defined as the set of all invariants of some algebra, which is a very natural way to define a constraint language for the CSP.
Thus, it is interesting to answer the following open questions.

\begin{problem}
Do Conjectures \ref{HubieConj} and \ref{HubieConjTwo} hold for 
$\Gamma = \{H\}$, where $H$ is a binary relation?
\end{problem}

\begin{problem}
Do Conjectures \ref{HubieConj} and \ref{HubieConjTwo} hold for 
$\Gamma$ consisting of binary relations?
\end{problem}

\begin{problem}
Do Conjectures \ref{HubieConj} and \ref{HubieConjTwo} hold for 
$\Gamma$ that is closed under projections (adding existential quantifiers)?
\end{problem}

For an algebra $\mathbf A = (A;f_1,f_{2},\dots)$
by $\Inv(\mathbf A)$ we denote the set of all relations
preserved by every operation $f_{i}$ of an algebra.

\begin{problem}
Do Conjectures \ref{HubieConj} and \ref{HubieConjTwo} hold for 
$\Gamma$ such that $\Gamma=\Inv(\mathbf A)$
for some algebra $\mathbf A$?
\end{problem}

The paper is organized as follows.
In Section \ref{FirstProofSection} and Section \ref{SecondProofSection} we give two different proofs 
of Theorem \ref{NoRainbowTHM}.
Both proofs are based on the reduction from an NP-hard CSP problem, 
the first reduction is easier in terms of complexity (just three variables for every variable of the original instance) but the idea of the reduction is hidden there.
We hope the second reduction is better for understanding
but there we create nine variables for every variable of the original instance.
In Section \ref{CounterExampleSection} we 
prove Theorem \ref{CEthm}.

\section{No-Rainbow problem is NP-Hard (the first proof)}
\label{FirstProofSection}

Here we present our first proof of the fact that $\SCSP(\{N\})$ is NP-hard.
By $[n]$ we denote the set $\{1,2,\dots,n\}$.
Let $\Gamma_{0} = \{x=y\vee z=0,x=y\vee z=1, x=0, x=1\}$ be a constraint language 
on $\{0,1\}$.
The problem $\CSP(\Gamma_{0})$ is known to be NP-Hard \cite{Schaefer}.
Hence, to prove Theorem \ref{NoRainbowTHM}
it is sufficient to 
prove the following theorem.

\begin{thm}\label{NoRainbowIsHard}
$\CSP(\Gamma_{0})$ can be polynomially reduced to $\SCSP(\{N\})$.
\end{thm}

Let us define the reduction. Let 
$\mathcal I$ be an instance of $\CSP(\Gamma_{0})$.
Let us build an instance $\mathcal J$ of $\SCSP(\{N,=\})$
such that $\mathcal I$ holds if and only if $\mathcal J$
has a surjective solution.
Note that 
the problems 
$\SCSP(\{N,=\})$
and 
$\SCSP(\{N\})$ are equivalent because we can always propagate out the equalities.

The idea is to introduce constraints (see 
$\mathcal C_{1},\mathcal C_{2},\dots, \mathcal C_{10}$ below) such that 
three concrete variables take all the values from the domain 
in any surjective solution. 
Below such variables are $x,x',y'$. 
Other variables still have some freedom, for instance $x_{1},\dots,x_{n}$ can be chosen 
freely from $\{x,x'\}$.
These variables will be used to encode 
the variables $u_{1},\dots,u_{n}$ of the instance $\mathcal I$.
By adding new constraints 
(see 
$\mathcal C_{11},\mathcal C_{12},\mathcal C_{13}, \mathcal C_{14}$ below)
we can encode constraints 
of $\CSP(\Gamma_{0})$ on the variables 
$x_{1},\dots,x_{n}$. A formal construction is below.


\textbf{Construction}.
Let $u_{1},\dots,u_{n}$ be the variables of $\mathcal I$.

Choose variables $x,x',y',x_{1},\dots,x_{n},y_{1},\dots,y_{n},
z_{1},\dots,z_{n}$.
We define 14 sets of constraints:

\begin{itemize}
\item $\mathcal C_{1} = \{N(x,x_{i},y_{i})\mid i\in[n]\}$,

\item $\mathcal C_{2} = \{N(x',z_{i},y_{i})\mid i\in[n]\}$,

\item $\mathcal C_{3} = \{N(y',z_{i},x_{i})\mid i\in[n]\}$,

\item $\mathcal C_{4} = \{N(t_1,t_2,t_3)\mid 
t_{1},t_{2},t_{3}\in\{x,x',x_{1},\dots,x_{n}\}\}$,

\item $\mathcal C_{5} = \{N(t_1,t_2,t_3)\mid 
t_{1},t_{2},t_{3}\in\{x,y',y_{1},\dots,y_{n}\}\}$,

\item $\mathcal C_{6} = \{N(t_1,t_2,t_3)\mid 
t_{1},t_{2},t_{3}\in\{x',y',z_{1},\dots,z_{n}\}\}$,

\item $\mathcal C_{7} = \{N(x_{i},y_{j},z_{i})\mid 
i,j\in[n]\}$,

\item $\mathcal C_{8} = \{N(x_{i},y_{j},z_{j})\mid 
i,j\in[n]\}$,

\item $\mathcal C_{9} = \{N(x, x_{i},z_{i})\mid 
i\in[n]\}$,

\item $\mathcal C_{10} = \{N(x, y_{i},z_{i})\mid 
i\in[n]\}$,

\item $\mathcal C_{11} = \{N(x_{i}, x_{j},y_{k})\mid 
(u_{i}=u_{j}\vee u_{k} = 0)\in \mathcal I\}$,

\item $\mathcal C_{12} = \{N(y_{i}, y_{j},x_{k})\mid 
(u_{i}=u_{j}\vee u_{k} = 1)\in \mathcal I\}$,

\item $\mathcal C_{13} = \{x=x_{i}\mid 
(u_{i}=1)\in \mathcal I\}$,

\item $\mathcal C_{14} = \{x=y_{i}\mid 
(u_{i}=0)\in \mathcal I\}$.
\end{itemize}

\begin{lem}
$\mathcal I$ has a solution if and only if 
$\mathcal J$ has a surjective solution, where 
$\mathcal J = \mathcal C_{1}\cup\dots\cup \mathcal C_{14}$.
\end{lem}

\begin{proof}
Let us show both implications.

$\mathcal I\in\CSP(\Gamma_{0})\Rightarrow 
\mathcal J\in\SCSP(\{N,=\}).$

Let $(b_{1},\dots,b_{n})$ be a solution of $\mathcal I$.
Put 
$x = 1$, $x' = 0$, $y'=2$, 
$x_{i} = b_{i}$, $y_{i} = b_{i}+1$,
$z_{i} = 2\cdot b_{i}$ for every $i\in[n]$.
Let us check that all the constraints are satisfied.
\begin{itemize}
\item $\mathcal C_{1}$
holds because $|\{1,b_{i},b_{i}+1\}|<3$ for any $b_{i}\in\{0,1\}$.

\item $\mathcal C_{2}$
holds because 
$|\{0,2\cdot b_{i},b_{i}+1\}|<3$ for any $b_{i}\in\{0,1\}$.

\item $\mathcal C_{3}$
holds because 
$|\{2,2\cdot b_{i},b_{i}\}|<3$ for any $b_{i}\in\{0,1\}$.

\item $\mathcal C_{4}$
holds because 
$x,x',x_{1},\dots,x_{n}$ are from the set $\{0,1\}$.

\item $\mathcal C_{5}$
holds because 
$x,y',y_{1},\dots,y_{n}$ are from the set $\{1,2\}$.

\item $\mathcal C_{6}$
holds because 
$x',y',z_{1},\dots,z_{n}$ are from the set $\{0,2\}$.

\item $\mathcal C_{7}$
holds because $|\{b_{i},b_{j}+1,2\cdot b_{i}\}|<3$ for any $b_{i},b_{j}\in\{0,1\}$.

\item $\mathcal C_{8}$
holds because $|\{b_{i},b_{j}+1,2\cdot b_{j}\}|<3$ for any $b_{i},b_{j}\in\{0,1\}$.

\item $\mathcal C_{9}$
holds because $|\{1,b_{i},2\cdot b_{i}\}|<3$ for any $b_{i}\in\{0,1\}$.

\item $\mathcal C_{10}$
holds because $|\{1,b_{i}+1,2\cdot b_{i}\}|<3$ for any $b_{i}\in\{0,1\}$.

\item $\mathcal C_{11}$
holds because 
each constraint is equivalent 
to 
$|\{b_{i},b_{j},1+b_{k}\}|<3$ 
and 
equivalent to 
$(b_{i}=b_{j})\vee b_{k}=0$.

\item $\mathcal C_{12}$
holds because 
each constraint is equivalent 
to 
$|\{b_{i}+1,b_{j}+1,b_{k}\}|<3$ 
and 
equivalent to 
$(b_{i}=b_{j})\vee b_{k}=1$.

\item $\mathcal C_{13}$ holds because  
$x_{i} =x= 1$ whenever $b_{i} = 1$. 

\item $\mathcal C_{14}$ holds because  
$y_{i} =x=1$ whenever $b_{i} = 0$. 
\end{itemize}

$\mathcal J\in\SCSP(\{N,=\})\Rightarrow 
\mathcal I\in\CSP(\Gamma_{0}).$

Choose a surjective solution of $\mathcal J$. 
By $\mathcal C_{4}$ we can choose $a,b\in\{0,1,2\}$ such that 
$x=a$ and 
$\{x,x',x_{1},\dots,x_{n}\}\subseteq \{a,b\}$. 
Since the solution is surjective, there should be 
an element in the solution equal to $c\in\{0,1,2\}\setminus\{a,b\}$.
Consider 5 cases. 

\textbf{Case 1}. Assume that $y'=c$ and $x'=b$.
Note that by the definition of $N$ for any permutation $\sigma$ on $\{0,1,2\}$
$\sigma$ applied to a solution of an instance of $\SCSP(\{N\})$ gives a solution of the instance. Therefore, without loss of generality 
we may assime that 
$x=1,x' = 0, y' = 2$ in our solution.
By $\mathcal C_{4}$,
$\mathcal C_{5}$,
$\mathcal C_{6}$
we know that 
\begin{align*}
    \{x,x',x_{1},\dots,x_{n}\}&=\{0,1\},\\
    \{x,y',y_{1},\dots,y_{n}\}&=\{1,2\},\\
    \{x',y',z_{1},\dots,z_{n}\}&=\{0,2\}.
\end{align*}
Let us show that $y_{i} = x_{i}+1$.
If $x_{i} =0$ and $y_{i} = 2$, then we get a contradiction 
with $\mathcal C_{1}$.

If $x_{i} =y_{i} = 1$, then 
by $\mathcal C_{2}$ we have $z_{i}\neq 2$, 
by $\mathcal C_{3}$ we have $z_{i}\neq 0$,
which implies $z_{i}=1$ and contradicts $\mathcal C_6$.

Let us show that 
$(u_{1},\dots,u_{n}) = (x_{1},\dots,x_{n})$ is 
a solution of $\mathcal I$.

$\mathcal C_{11}$ guarantees that 
all the constraints of the form 
$(u_{i}=u_{j}\vee u_{k} = 0)$ hold, 
$\mathcal C_{12}$ guarantees that 
all the constraints of the form 
$(u_{i}=u_{j}\vee u_{k} = 1)$ hold, 
$\mathcal C_{13}$ guarantees that 
all the constraints of the form 
$u_{i}=1$ hold,
$\mathcal C_{14}$ guarantees that 
all the constraints of the form 
$u_{i}=0$ hold.
Thus, we proved that it is a solution.

\textbf{Case 2}. Assume that $y'=c$ and $x'=a$.
By $\mathcal C_{5}$
we have 
$\{x,y',y_1,\dots,y_n\}=\{a,c\}$.
By $\mathcal C_{6}$
we have 
$\{x',y',z_1,\dots,z_n\}=\{a,c\}$.
Since the solution is surjective, there should be $i$ 
such that $x_{i} = b$.
By $\mathcal C_{9}$ we have $z_{i}\neq c$,
by $\mathcal C_{3}$ we have $z_{i}\neq a$. Contradiction.

\textbf{Case 3}. Assume that $y'\neq c$ and $y_{i}=z_{i}=c$ for some $i$.
By $\mathcal C_{5}$
we have 
$\{x,y',y_1,\dots,y_n\}=\{a,c\}$, 
hence $y' = a$.
By $\mathcal C_{6}$
we have 
$\{x',y',z_1,\dots,z_n\}=\{a,c\}$, 
hence 
$x' = a$.
Since the solution is surjective, there should be $j$ 
such that $x_{j} = b$.
By $\mathcal C_{9}$, we have $z_{j} = a$.
By $\mathcal C_{7}$ 
we have $y_{\ell} = a$ for every $\ell$.
Contradiction.

\textbf{Case 4}. Assume that $y'\neq c$ and $y_{i}=c, z_{i}\neq c$ for some $i$.
By $\mathcal C_{5}$
we have 
$\{x,y',y_1,\dots,y_n\}=\{a,c\}$, 
hence $y' = a$.
By $\mathcal C_{1}$, we have 
$x_{i} = a$.
By $\mathcal C_{8}$ we obtain 
$z_{i}\in\{a,c\}$ and since $z_{i}\neq c$ we have
$z_{i} = a$.
Then by $\mathcal C_{8}$ we obtain 
$x_{\ell}=a$ for every $\ell$.
If $z_{j} = b$ for some $j$, then 
we get a contradiction with 
$\mathcal C_{7}$ applied to 
$(x_{j},y_{i},z_{j})$.
Since the solution is surjective, the only remaining option is 
$x'=b$, which contradicts 
$\mathcal C_{2}$.

\textbf{Case 5}. Assume that 
$\{x,y',y_1,\dots,y_n\}\subseteq\{a,b\}$.
Since the solution is surjective, 
there exists $i$ such that 
$z_{i} = c$.
By $\mathcal C_{9}$ and $\mathcal C_{10}$
we have $x_{i} = y_{i} = a$.
By $\mathcal C_{7}$ we have 
$y_{\ell} = a$ for every $\ell$,
by $\mathcal C_{8}$ we have 
$x_{\ell} = a$ for every $\ell$.
By $\mathcal C_{2}$ we have $x'\neq b$, 
by $\mathcal C_{3}$ we have $y'\neq b$.
Therefore $x'=y' = a$
and by $\mathcal C_{6}$
we obtain 
$\{x',y',z_1,\dots,z_n\}=\{a,c\}$.
Hence, none of the variables can be equal to $b$. Contradiction.
\end{proof}

\section{No-Rainbow problem is NP-Hard (the second proof)}\label{SecondProofSection}
Here we present another proof of the fact that 
the No-Rainbow problem is NP-Hard.
By $NAE_{3}$ we denote the ternary relation 
on $\{0,1\}$ consisting of all tuples 
but $(0,0,0)$ and $(1,1,1)$.
The idea is to encode an instance 
of $\CSP(\{NAE_{3}\})$ using the relation $N$. 
We assign a binary operation 
on $A = \{0,1,2\}$ to every variable of the 
instance of $\CSP(\{NAE_{3}\})$ and then encode each binary operation with nine variables on $\{0,1,2\}$.
Our plan is to write conditions that guarantee 
that every binary operation depends essentially only on one variable.
Later we interpret the dependence on the first variable as $0$ and 
the dependence on the second variable as $1$.

It is known from \cite{Schaefer} that 
$\CSP(\{NAE_{3}\})$ is NP-hard. 
Hence, to prove Theorem \ref{NoRainbowTHM}
it is sufficient to reduce $\CSP(\{NAE_{3}\})$
to $\SCSP(\{N\})$.
We start with a few auxiliary facts.
For a relation $\rho$ by $\Pol(\rho)$ we denote the set of all polymorphisms of $\rho$.
We say that an operation $f$
\emph{depends essentially only on one variable} if
$f(x_1,\dots,x_{n}) = g(x_{i})$ for some $i$ and a unary operation $g$, 
that is all variables but $i$-th are dummy.
We will need the following properties of $\Pol(N)$.

\begin{lem}[Section 5.2.6 in \cite{lau}]\label{SlupetskiClone}
Suppose $f\in \Pol(N)$ then 
\begin{enumerate}
    \item $f$ depends essentially only on one variable, or
    \item $|\Image(f)|<3$, that is, $f$ never returns some value
    $a\in A$.
\end{enumerate}
\end{lem}

We will encode each binary operation $f$ on $A$ via 
9 variables on $A$, which we denote by
\begin{align*}
f(0,0), f(0,1),f(0,2),
f(1,0), &f(1,1),f(1,2),\\
&f(2,0), f(2,1),f(2,2).
\end{align*}
For a tuple $\alpha$ by $\alpha(i)$ we denote the $i$-th element of this tuple. 
If we just write the definition of the fact that 
$f$ is a polymorphism of $N$ we get the following lemma.

\begin{lem}\label{PolWrittenAsConj}
$f\in \Pol(N)$ if and only if 
$$\bigwedge_{\alpha,\beta\in N} 
    N(f(\alpha(1),\beta(1)),f(\alpha(2),\beta(2)),f(\alpha(3),\beta(3))).$$
\end{lem}

Thus, the condition $f\in\Pol(N)$ can be expressed as a 
conjunction of relations $N$. 

Now we are ready to prove the main theorem of this subsection.

\begin{thm}
$\CSP(\{NAE_{3}\})$ can be polynomially reduced to $\SCSP(\{N\})$.
\end{thm}




Let us show how to encode $\CSP(\{NAE_{3}\})$ as 
$\SCSP(\{N\})$.
Consider an instance $\mathcal I$ of $\CSP(\{NAE_{3}\})$.
Let $x_{1},\dots,x_{n}$ be the variables 
of $\mathcal I$
and $T$ be the set of all triples 
$(i,j,k)$ such that 
$NAE_{3}(x_{i},x_{j},x_{k})$ appears in the instance.

As we mentioned earlier, we assign a binary operation $f_{i}$ on $A$ to each variable $x_{i}$, then we encode 
each operation with nine variables on $A$, which we denote by 
$f_{i}(0,0),f_{i}(0,1),f_{i}(0,2),$ 
$f_{i}(1,0)$, $f_{i}(1,1),$ $f_{i}(1,2),$ 
$f_{i}(2,0),$ $f_{i}(2,1),$ $f_{i}(2,2)$.

We want $f_{i}$ to depend only on the first variable
if $x_{i}=0$ and only on the second variable if $x_{i}=1$.
By $\mathcal I'$ we denote the following instance:
\begin{align*}
&\bigwedge\limits_{i\in [n]}(f_{i}\in\Pol(N))
\bigwedge\limits_{i\in [n], a\in A}(f_{1}(a,a)=f_{i}(a,a))\\
&\bigwedge\limits_{(i,j,k)\in T}
N(f_{i}(0,1),f_{j}(1,2),f_{k}(2,0))\\
 &\bigwedge\limits_{i,j\in[n],a,b,c\in A}N(f_{i}(a,b),f_{i}(c,b),f_{j}(a,c))\\
  &\;\;\;\;\;\bigwedge\limits_{i,j\in[n],a,b,c\in A}N(f_{i}(b,a),f_{i}(b,c),f_{j}(a,c))
\end{align*}
Note that by Lemma \ref{PolWrittenAsConj} the first conjunction 
can be written as a conjunction of relations $N$.
Hence, this instance can be viewed as an instance of 
$\SCSP(\{N\})$ because we use only equalities and the relation $N$
(the equalities can be propagated out).

\begin{lem}
$\mathcal I$ has a solution if and only if 
$\mathcal I'$ has a surjective solution.
\end{lem}
\begin{proof}

$\mathcal I\Rightarrow \mathcal I'$.
Suppose we have a solution 
$(x_{1},\dots,x_{n})$ of $\mathcal I$. To get a surjective solution of $\mathcal I'$ it is sufficient to put
$f_{i}(x,y) = x$ if $x_{i} =0$ and
$f_{i}(x,y) = y$ if $x_{i} =1$.
Since $f_{i}(0,0) = 0$, 
$f_{i}(1,1) = 1$,
and
$f_{i}(2,2) = 2$ the solution is surjective.
Let us prove that all the above conjunctions hold.
The first conjunction holds since projections preserve
$N$. The second conjunction is trivial. The third conjunction holds because 
$(x_{1},\dots,x_{n})$ is a solution of $\mathcal I$
and therefore the operations $f_{i},f_{j},f_{k}$ cannot depend 
on the same variables. 
Let us check the fourth conjunction.
If $x_{i}=0$ then $f_{i}(a,b),f_{i}(c,b),
f_{j}(a,c)\in\{a,c\}$,
if $x_{i}=1$ then $f_{i}(a,b),f_{i}(c,b),
f_{j}(a,c)\in\{b,f_{j}(a,c)\}$.
The remaining conjunction can be verified in the same way. Thus, we get a surjective solution of $\mathcal I'$.

$\mathcal I'\Rightarrow \mathcal I$.
Suppose we have a surjective solution 
$(f_{1},\dots,f_{n})$ of $\mathcal I'$.
Let $g(x) := f_{1}(x,x)$.
By the second conjunction we have 
$f_{i}(x,x) = g(x)$ for every $i\in[n]$.

Assume that $|\Image(g)|=3$. By the first conjunction and Lemma \ref{SlupetskiClone} 
each $f_{i}$ depends essentially only on one variable.  
Assign $x_{i}=0$ if $f_{i}$ depends essentially on the first variable, 
and $x_{i}=1$ if $f_{i}$ depends essentially on the second variable. 
The third conjunction guarantees that 
$f_{i}, f_{j},$ and $f_{k}$ cannot depend on the same coordinate for each $(i,j,k)\in T$.
Indeed, if three operations $f_{i}, f_{j}, f_{k}$ 
depend only on the first variables then 
we get 
$$N(f_{i}(0,1),f_{j}(1,2),f_{k}(2,0)) = N(g(0),g(1),g(2)),$$ which does not hold because $g$ is a bijection. Thus, 
for each $(i,j,k)\in T$
we have $NAE_{3}(x_{i},x_{j},x_{k})$.
Hence, we defined a solution of $\mathcal I$.

Assume that $|\Image(g)|=2$. 
If $|\Image(f_{i})| =3$ for some $i$ then 
by Lemma \ref{SlupetskiClone} 
$f_{i}$ depends essentially only on one variable.
This means that 
$\Image(f_{i}) = \Image(g)$ for every $i\in[n]$, which contradicts the surjectivity of 
the solution.

Assume that $\Image(g)=\{d\}$. 
Assume that $|\Image(f_{k})|=3$ for some $k\in[n]$.
Then by Lemma~\ref{SlupetskiClone} $f_{k}$ depends essentially only on one variable and $\Image(f_{k})=\Image(g) = \{d\}$, which gives us a contradiction.
Thus, $|\Image(f_{k})|<3$ for every $k\in[n]$.
Since the solution is surjective, the remaining 
two values from $A\setminus\{d\}$ should appear in the images of some functions $f_{i}$ and $f_{j}$.
To get a contradiction we use the last two conjunctions. 
The fifth conjunction is obtained by a permutation of variables in 
$f_{i}$
from the fourth conjunction.
Therefore, without loss of generality
we consider 
$f_{i},f_{j}$ and $a,b,c\in A$ such that 
$\{d,f_{i}(a,b),f_{j}(a,c)\}=A$.
By the fourth conjunction we have 
$N(f_{i}(a,b),f_{i}(c,b),f_{j}(a,c))$, 
which together with $\Image(f_{i}) = \{f_{i}(a,b),d\}$
implies 
$f_{i}(c,b)= f_{i}(a,b)$. 
Then by the fifth conjunction we get 
$N(f_{j}(a,c), f_{j}(a,b), f_{i}(c,b))$
which together with $\Image(f_{j}) = \{f_{j}(a,c),d\}$
implies 
$f_{j}(a,b) = f_{j}(a,c)$.
Then by the fourth conjunction we have 
$(f_{i}(a,b),f_{i}(b,b),f_{j}(a,b))
= (f_{i}(a,b),d,f_{j}(a,c)) \in N$, which contradicts our assumption.
\end{proof}

%% file: counterexample.tex
\section{Counter-example}\label{CounterExampleSection}

Here we will prove Theorem \ref{CEthm} from Section 
\ref{MainResultsSection}.

Recall that
$R = 
\begin{pmatrix}
0&0&0&0&0&0\\
2&2&2&1&1&1\\
2&1&1&1&1&1\\
1&2&2&1&1&1\\
0&1&2&0&1&2
\end{pmatrix}$,
$R'=\begin{pmatrix}
2&2&1\\
2&1&1\\
1&2&1
\end{pmatrix}$.

\begin{lem}
$\CSP(\{R,\{0\},\{1\},\{2\}\})$ is NP-hard.
\end{lem}

\begin{proof}
Note that 
$R'$ is the projection of $R$ onto the coordinates 2, 3, and 4.
We can check that 
$R'$ is not preserved by any idempotent WNU on $\{1,2\}$.
By 
Theorem \ref{FederVardiConj},
$\CSP(\{R',\{1\},\{2\}\})$ is NP-hard.
Therefore, 
$\CSP(\{R,\{0\},\{1\},\{2\}\})$ is NP-hard.
\end{proof}

\begin{lem}
$\SCSP(\{R'\})$ is NP-hard.
\end{lem}
\begin{proof}
$\SCSP(\{R'\})$ can be viewed  as the SCSP on the 
two-element set $\{1,2\}$ (we just add a variable that never appears for 0 to be in a solution), which is known to be 
NP-Hard \cite{SCSPforTwo1997,SCSPforTwoBook}.
\end{proof}

\begin{lem}
$\SCSP(\{R\})$ is solvable in polynomial time.
\end{lem}

\begin{proof}
The idea of the algorithm is very simple. 
To prevent our instance from 
having a trivial surjective solution
we have to restrict the fifth variable of 
every appearance of $R$ to a smaller domain.
After we did this, we can replace $R$ by a conjunction of easier relations.
Thus, we show that our instance has a trivial surjective solution unless all appearances of the relation $R$ could be replaced by easier relations. 

Let 
$\sigma = \{(1,1),(2,2)\}$. 
We will prove a stronger claim that 
$\SCSP(\Gamma)$ is solvable in polynomial time for 
$\Gamma = \{R,\sigma,\{0\},\{1\},\{1,2\}\}$.
Consider an instance $\mathcal I$ of $\SCSP(\Gamma)$.

First, we want to classify every occurrence of a variable in the instance. 
We say that an occurrence is of \emph{the first type} if 
the projection of the constraint onto this variable is 
$\{0\}$,
we say that 
an occurrence is of \emph{the second type} if 
the projection of the constraint onto this variable is 
a subset of $\{1,2\}$.
In all other cases we say that it is an occurrence
of \emph{the third type}.
For example 
in $R(x_1,x_2,x_3,x_4,x_5)$
the variable $x_1$ is of the first type, 
the variables $x_2,x_3,x_4$ are of the second type 
and $x_5$ is of the third type.
In $\sigma$ both variables are of the second type.
Note that the third type appears only in the relation $R$.

Second, we want all the occurrences of 
each variable to be of one type. We do the following:
\begin{itemize}
    \item If a variable occurs in the first and the second types then we return ``No solutions''.    
    \item If a variable occurs in the first and third types then it should appear at the last position of 
    $R$. Then we replace the relation $R$ by $\sigma$, $\{0\}$ and $\{1\}$ using the following equation
    \begin{align*}
R(x_1,x_2,x_3,x_4,x_{5})\wedge  
(x_{5}=0)
=&\\
(x_{1}=0)\wedge \sigma(x_2,x_3)\wedge &(x_{4}=1)\wedge
(x_{5}=0)
    \end{align*}
    \item If a variable occurs in the second and third types then it should appear at the last position of 
    $R$. Then we replace the relation $R$ by $\sigma$, $\{0\}$ and $\{1,2\}$ using the following equation%
\begin{align*}
R(x_1,x_2,x_3,x_4,x_{5})\wedge  
(x_{5}\in \{1,2\})
&=\\
(x_{1}=0)\wedge \sigma(x_2,x_4)\wedge (x_{3}=&1)\wedge 
(x_{5}\in \{1,2\})
\end{align*}
\end{itemize}
Finally we get an instance with the same solution set such that all the occurrences of each variable have the same type. Here we may have two cases: 

Case 1. The instance $\mathcal I$ does not contain $R$ at all. Such an instance is in fact trivial because it contains only equality relation on $\{1,2\}$ and unary relations. Hence, we can check whether it has a surjective solution in polynomial time. To avoid a formal explanation of how to do this, we can reduce such SCSP to CSP.
Since $\CSP(\{\sigma,\{0\},\{1\},\{2\},\{1,2\}\})$ can be solved in polynomial time \cite{BulatovForThree}, 
by Lemma \ref{SCSPTOCSP}
$\SCSP(\{\sigma,\{0\},\{1\},\{1,2\}\})$ is also solvable in polynomial time. 
We use this algorithm for our instance.


Case 2. The instance $\mathcal I$ contains $R$, which means that the instance has occurrences of variables of all three types (since all types appear in $R$).
To get a surjective solution it is sufficient to 
send the variables of the first type to 0, the variables of the second type to 1, and the variables of the third type to 2.
Since $R$ holds on $(0,1,1,1,2)$ and $\sigma$ holds on 
$(1,1)$, this is in fact a solution.
\end{proof}
